\documentclass[a4paper,USenglish]{lipics-v2016-arxiv}

\usepackage{microtype}

\title{Optimal Data Reduction for Graph Coloring Using Low-Degree Polynomials\footnote{This work was supported by NWO Veni grant ``Frontiers in Parameterized Preprocessing'' and NWO Gravitation grant ``Networks''. An extended abstract appeared in the Proceedings of the 12th International Symposium on Parameterized and Exact Computation, IPEC 2017.}}

\author[1]{Bart M.\,P. Jansen}
\author[2]{Astrid Pieterse}
\affil[1]{Eindhoven University of Technology\\
  P.O. Box 513, Eindhoven, The Netherlands\\
  \texttt{b.m.p.jansen@tue.nl}}
\affil[2]{Eindhoven University of Technology\\
  P.O. Box 513, Eindhoven, The Netherlands\\
 \texttt{a.pieterse@tue.nl}}
\authorrunning{B.\,M.\,P. Jansen and A. Pieterse} 

\Copyright{Bart M.\,P. Jansen and Astrid Pieterse}


\EventEditors{}
\EventNoEds{0}
\EventLongTitle{arXiv-version}
\EventShortTitle{arXiv-version}
\EventAcronym{arXiv-version}
\EventYear{2018}
\EventDate{}
\EventLocation{}
\EventLogo{}
\SeriesVolume{}
\ArticleNo{}

\newcounter{NameOfTheNewCounter}

\usepackage{xspace}
\theoremstyle{plain}

\newtheorem{myclaim}[theorem]{Claim}
\newtheorem{observation}[theorem]{Observation}
\newtheorem{rrule}[NameOfTheNewCounter]{Reduction rule}
\newcommand{\vectspan}{\ensuremath{\operatorname{span}_2}}
\newcommand{\somevector}{\operatorname{vect}}
\let\plainqed\qedsymbol
\newcommand{\claimqed}{$\lrcorner$}
\newenvironment{claimproof}{\begin{proof}\renewcommand{\qedsymbol}{\claimqed}}{\end{proof}\renewcommand{\qedsymbol}{\plainqed}}

\newcommand{\F}{F\xspace}

\newcommand{\pname}[1]{\textsc{#1}\xspace}
\newcommand{\notcontainment}{\ensuremath{\mathsf{NP \not\subseteq coNP/poly}}\xspace}
\newcommand{\NP}{\ensuremath{\mathsf{NP}}\xspace}

\newcommand{\eqvr}[0]{\ensuremath{\mathcal{R}}\xspace}
\newcommand{\containment}{\ensuremath{\mathsf{NP  \subseteq coNP/poly}}\xspace}
\newcommand{\q}{t'}
\newcommand{\N}{\ensuremath{\mathbb{N}}\xspace}
\newcommand{\Q}{\ensuremath{\mathcal{Q}}\xspace}
\newcommand{\OR}{\textsc{or}}
\newcommand{\Oh}{\mathcal{O}}
\newcommand{\FPT}{\ensuremath{\mathsf{FPT}}\xspace}
\newcommand{\true}{\emph{true}\xspace}

\newcommand{\ThreeColoring}{\pname{$3$-Coloring}\xspace}
\newcommand{\newgadget}{\ensuremath{\text{blocking-gadget}}\xspace}
\newcommand{\newgadgets}{\ensuremath{\text{blocking-gadgets}}\xspace}
\newcommand{\vect}[1]{\ensuremath{\mathbf{#1}}\xspace}

\newcommand{\qcoloring}{\textsc{$q$-Coloring}\xspace}
\newcommand{\todo}[1][]{%
  \ifx/#1/%
    \textcolor{red}{TODO!}%
  \else%
    \textcolor{red}{todo: #1}%
  \fi%
}

\newcommand{\defproblem}[3]{\par
 \vspace{3mm}
\noindent\fbox{
 \begin{minipage}{0.96\textwidth}
 \begin{tabular*}{\textwidth}{@{\extracolsep{\fill}}lr} #1 &  \vspace{1mm} \\ \end{tabular*}
{\textbf{Input:}} #2
  \vspace{1mm}\\%
 {\textbf{Question:}} #3
 \end{minipage}
 }
 \vspace{3mm}\par
}

\begin{document}

\maketitle

\begin{abstract}
The theory of kernelization can be used to rigorously analyze data reduction for graph coloring problems. Here, the aim is to reduce a \textsc{$q$-Coloring} input to an equivalent but smaller input whose size is provably bounded in terms of structural properties, such as the size of a minimum vertex cover. In this paper we settle two open problems about data reduction for \textsc{$q$-Coloring}.

First, we obtain a kernel of bitsize $\mathcal{O}(k^{q-1}\log{k})$ for \textsc{$q$-Coloring parameterized by Vertex Cover} for any $q\geq 3$. This size bound is optimal up to~$k^{o(1)}$ factors assuming~$\mathsf{NP \not\subseteq coNP/poly}$, and improves on the previous-best kernel of size~$\mathcal{O}(k^q)$. We generalize this result for deciding $q$-colorability of a graph $G$, to deciding the existence of a homomorphism from $G$ to an arbitrary fixed graph $H$. Furthermore, we can replace the parameter vertex cover by the less restrictive parameter twin-cover. We prove that \textsc{$H$-Coloring parameterized by Twin-Cover} has a kernel of size $\Oh(k^{\Delta(H)}\log k)$.

Our second result shows that \textsc{$3$-Coloring} does not admit non-trivial sparsification: assuming $\mathsf{NP \not\subseteq coNP/poly}$, the parameterization by the number of vertices~$n$ admits no (generalized) kernel of size~$\mathcal{O}(n^{2-\varepsilon})$ for any~$\varepsilon > 0$. Previously, such a lower bound was only known for coloring with~$q \geq 4$ colors.

\end{abstract}

\section{Introduction}
\label{intro}

The \pname{$q$-Coloring} problem asks whether the vertices of a graph can be properly colored using~$q$ colors. It is one of many colorability problems on graphs that have been widely studied.
Since these are often NP-hard, they are good candidates to study from a parameterized perspective~\cite{CyganFKLMPPS15,DowneyF13}. Here we use additional parameters, other than the size of the input, to describe the complexity of the problem. In this paper we study preprocessing algorithms (called kernelizations or kernels) that aim to reduce the size of an input graph in polynomial time, without changing its colorability status.

The natural choice for a parameter for \pname{$q$-Coloring} is the number of colors~$q$. However, since even \pname{$3$-Coloring} is NP-hard, this parameter does not give interesting results. Therefore the problem is studied using different parameters, that often try to capture the complexity of the input graph.
For example, Fiala et. al. \cite{FialaGK11Parameterized} compared the parameterized complexity of several coloring problems when parameterized by vertex cover, to the complexity when parameterized by treewidth. Jansen and Kratsch \cite{JansenK2013Data} studied  graph coloring when parameterized by a hierarchy of different parameters.

In this earlier work \cite{JansenK2013Data}, Jansen and Kratsch provided a kernel for \pname{$q$-Coloring parameterized by Vertex Cover}
with $\Oh(k^q)$ vertices that can be encoded in~$\Oh(k^q)$ bits. Furthermore they showed that for $q \geq 4$, a kernel of bitsize $\Oh(k^{q-1-\varepsilon})$ is unlikely to exist. Unfortunately, these bounds left a gap of a factor $k$ and it remained  unclear whether the upper or the lower bound had to be strengthened. As our first main result, we manage to close this gap by improving the kernel.

To obtain this improvement, we can use a recent result by the current authors  \cite{JansenP2016Optimal} about the kernelization of constraint satisfaction problems when parameterized by the number of variables. A non-trivial data reduction can be achieved when the constraints are given by equalities of low-degree polynomials on boolean variables. The size of the resulting instance then depends on the maximum degree of the given polynomials. Suppose now we are given a \pname{$3$-Coloring} instance $G$ with vertex cover $S$ and let $I = V(G)\setminus S$ be the corresponding independent set. One can think of each vertex $v \in I$ as a constraint of the form ``my neighbors use at most $2$ different colors'', such that a remaining color can be used to color $v$. We write these constraints as polynomial equalities  and apply our previous result to find out which ones are redundant. Since vertices of the independent set can be colored independently, a vertex that corresponds to a redundant constraint can be removed from $G$, without changing the $3$-colorability of $G$. We can apply this idea to obtain a kernel for \pname{$q$-Coloring parameterized by Vertex Cover}. The key technical step is to build a polynomial of degree~$q-1$ that captures the desired constraints.

In this paper, we further generalize the problem by studying the \pname{$H$-Coloring} problem. The problem asks for a given graph $G$ and fixed graph $H$, whether there exists a homomorphism $f \colon V(G) \rightarrow V(H)$ such that $\{u,v\} \in E(G) \Rightarrow \{f(u),f(v)\} \in E(H)$. Instead of using the size of a vertex cover as the parameter, we use a smaller parameter called twin-cover \cite{Ganian15ImprovingVC}.  We show  in Theorem \ref{thm:kernel} that \pname{$H$-Coloring} parameterized by the size of a twin-cover has a kernel with $\Oh(k^{\Delta(H)})$ vertices and bitsize $\Oh(k^{\Delta(H)}\log k)$.
Since \pname{$q$-Coloring} is equivalent to \pname{$K_q$-Coloring} where $K_q$ is the clique on $q$ vertices, this result immediately gives a kernel of bitsize $\Oh(k^{q-1} \log k)$ for \pname{$q$-Coloring} parameterized by vertex cover. This closes the gap with the lower bound for \pname{$q$-Coloring} up to $k^{o(1)}$ factors.

Often, when describing a kernel for a problem parameterized by a structural parameter like vertex cover, it is assumed that (an approximation of) the minimum vertex cover is given with the input \cite{Bougeret2017treedepthmodulator,Hols2017vckernel}. However, an interesting feature of our kernel for \pname{$H$-Coloring} is that it can be computed without knowing an (approximation of the) optimal twin-cover of the input graph. The fact that the graph has size-$k$ twin-cover is only used to analyze the size of the resulting kernel.

Our second main result concerns the parameterization by the number of vertices~$n$. The current authors showed in earlier work \cite{JansenP2016Sparsification} that for a number of graph problems it is impossible to give a kernel of size $\Oh(n^{2-\varepsilon})$, unless \containment. This implies that the number of edges cannot efficiently be reduced to a subquadratic amount without changing the answer, a task that is also known as sparsification. For example, \pname{$q$-Coloring} was shown to have no non-trivial sparsification for any $q \geq 4$, unless \containment. The case for $q=3$ remained open. One might think that \pname{$3$-Coloring} is so restrictive, that a $3$-colorable instance is likely to either be sparse, or have a very specific structure. Exploiting this structure could then allow for a non-trivial sparsification. In Theorem \ref{thm:threecoloring:lowerbound} we show that this is not the case: \pname{$3$-Coloring} allows no kernel of size $\Oh(n^{2-\varepsilon})$, unless \containment.

\subparagraph*{Related work.}

Hell and Ne\v{s}et\v{r}il showed that \pname{$H$-Coloring} is \NP-hard for any non-bipartite graph $H$ that has no self-loops \cite{HELL199092}. For a bipartite graph, the problem is equivalent to testing whether the input graph is bipartite, and thus polynomial-time solvable.  Chitnis et al. show that the problem of finding a smallest set $W \subseteq V(G)$ such that $G-W$ is $H$-list-colorable is \FPT when $H$ is $(C_6,P_6)$-free and bipartite, when parameterized by the size of $H$ together with the solution size \cite{Chitnis2017}. 

Ganian introduced \pname{Twin-Cover} as a new parameter \cite{Ganian15ImprovingVC} and gives relations to existing parameters. For example, a minimum twin-cover is not larger than a minimum vertex cover, but twin-cover is incomparable to treewidth. The paper also gives an FPT algorithm for \pname{Precoloring Extension} parameterized by the size of a twin-cover, and studies a number of other problems using this parameter.

Dell and Van Melkebeek showed that for any~$d \geq 3$, \pname{$d$-CNF-Satisfiability} with $n$ variables has no kernel of size $\Oh(n^{d-\varepsilon})$, unless \containment \cite{DellM14Satisfiability}.
Continuing this line of research, precise kernel lower bounds were shown for a variety of problems.  For example, it was shown that \pname{Vertex Cover} is unlikely to have a kernel of size  $\Oh(k^{2-\varepsilon})$ \cite{DellM14Satisfiability}, while a kernel with $\Oh(k^{2})$ edges and $\Oh(k)$ vertices is known.
Furthermore, the \pname{Point-Line cover} problem, which asks to cover a set of $n$ points in the plane with at most $k$ lines, was proven to have a tight kernel lower bound of size $\Oh(k^{2-\varepsilon})$ \cite{KratschPR16Point}, assuming \notcontainment .
Dell and Marx \cite{DellM12Kernelization} proved polynomial kernelization lower bounds for several packing problems. They showed how a table structure can help realize the reduction that is needed for such a lower bound. We will also use this table structure in our lower bound.

\section{Preliminaries}
%
To denote the set of numbers $1$ to $n$, we use the following notation: $[n] := \{i \in \mathbb{N} \mid 1 \leq i \leq n\}$.
For $x,y \in \mathbb{Z}$ we write~$x \equiv_2 y$ to denote that $x$ and $y$ are congruent modulo $2$. For a finite set $X$ and non-negative integer $k$, let $\binom{X}{k}$ be the collection of all subsets of $X$ of size exactly $k$ and let $\binom{X}{\leq k}$ be the collection of all subsets of $X$ of size at most $k$.

\subsection{Graphs}

All graphs considered in this paper are finite, simple, and undirected. In particular, this means that graphs do not have self-loops. A graph $G$ has vertex set $V(G)$ and edge set $E(G)$. For sets $X,Y \subseteq V(G)$, let $E_G(X,Y) := \{\{x,y\} \in E(G) \mid x \in X, y \in Y\}$ denote the edges with one endpoint in $X$ and one endpoint in $Y$.
Let $G[S]$ for $S\subseteq V(G)$ denote the subgraph of $G$ induced by $S$. For vertex set $X\subseteq V(G)$ we use $G-X: = G[V(G)\setminus X]$ to denote the result of removing the vertices in $S$ from $G$. For $F\subseteq E(G)$,  let $G \setminus F$ denote the result of removing all edges in $F$ from $G$.
Let $\Delta(G)$ denote the maximum degree of any vertex in $G$ and let $\omega(G)$ denote the size of a largest clique in $G$.

For a vertex $u \in V(G)$, let $N_G(u):= \{v \in V(G) \mid \{u,v\} \in E(G)\}$ denote its \emph{open neighborhood} and let $N_G[u] := N_G(u) \cup \{u\}$ denote its \emph{closed neighborhood}. For a vertex set $S \subseteq V(G)$, let $N_G(S):= \{v \in V(G)\setminus S \mid \{u,v\} \in E(G)\}$ denote its open neighborhood.

A \emph{vertex cover} of a graph $G$ is a set $S \subseteq V(G)$ such that each edge has at least one endpoint in $S$ (equivalently, $G-S$ is an \emph{independent set}).
We say vertices $u$ and $v \in V(G)$ are \emph{(true) twins} whenever $N_G[u] = N_G[v]$. Note that this relation is transitive. We say $X \subseteq V(G)$ is a \emph{twin-cover} \cite{Ganian15ImprovingVC} of $G$, if for every edge $\{u,v\} \in E(G)$, vertex $u \in X$, or $v \in X$, or $u$ and $v$ are twins.

A \emph{proper $q$-coloring} of $G$ is a function $f \colon V(G) \rightarrow [q]$ such that for all $\{u,v\} \in E(G) \colon f(u) \neq f(v)$.
Let $G$ and $H$ be graphs. We say that $G$ is \emph{$H$-colorable}  if there exists a function $f \colon V(G) \rightarrow V(H)$ such that for all $\{u,v\} \in E(G)$ it holds that $\{f(u),f(v)\} \in E(H)$. Such a function is also called a \emph{homomorphism} from $G$ to $H$. Note that $G$ has a homomorphism to $K_q$ (a clique on~$k$ vertices) if and only if $G$ is $q$-colorable.
In this paper, we will only consider \pname{$H$-Coloring} where $H$ has no self-loops and is not bipartite, as otherwise the problem is polynomial-time solvable. We will frequently use the following properties of $H$-colorings in the remainder of the paper.

\begin{observation} \label{obs:clique:hcoloring}
Let $S \subseteq V(G)$ such that $G[S]$ is a clique and let $f$ be a proper $H$-coloring of $G$. Define $X := \{f(v) \mid v \in S\}$. Then $H[X]$ is a clique in $H$ and  all vertices in $S$ receive a different color, so that $|S| = |X|$.
\end{observation}

\begin{observation} \label{obs:numcolors:delta}
Let $v \in V(G)$ and let $f$ be a proper $H$-coloring of $G$. Then the number of colors used to color $N_G(v)$ is bounded by $\Delta(H)$.
\end{observation}

\subsection{Parameterized complexity}

A \emph{parameterized problem} \Q is a subset of $\Sigma^* \times \mathbb{N}$, where $\Sigma$ is a finite alphabet. Let $\Q,\Q' \subseteq \Sigma^*\times\mathbb{N}$ be parameterized problems and let $h\colon\mathbb{N}\rightarrow\mathbb{N}$  be a computable function. A \emph{generalized kernel for \Q into \Q' of size $h(k)$} is an algorithm that, on input $(x,k) \in \Sigma^*\times\mathbb{N}$, takes time polynomial in $|x|+k$ and outputs an instance $(x',k')$ such that:
\begin{enumerate}
\item $|x'|$ and $k'$ are bounded by $h(k)$, and
\item $(x',k')\in\Q'$ if and only if $(x,k) \in\Q$.
\end{enumerate}
The algorithm is a \emph{kernel} for $\Q$ if $\Q = \Q'$. It is a \emph{polynomial (generalized) kernel} if $h(k)$ is a polynomial.  Since a polynomial-time reduction to an equivalent sparse instance yields a generalized kernel, a lower bound for the size of a generalized kernel can be used to prove the non-existence of sparsification algorithms.

We use the framework of cross-composition~\cite{BodlaenderJK14} to establish kernelization lower bounds, requiring the definitions of polynomial equivalence relations  and \OR-cross-compositions. We repeat them here for completeness:

\begin{definition}[{Polynomial equivalence relation, \cite[Def. 3.1]{BodlaenderJK14}}] \label{definition:eqvr}
An equivalence relation~\eqvr on~$\Sigma^*$ is called a \emph{polynomial equivalence relation} if the following conditions hold.
\begin{itemize}
\item There is an algorithm that, given two strings~$x,y \in \Sigma^*$, decides whether~$x$ and~$y$ belong to the same equivalence class in time polynomial in~$|x| + |y|$.
\item For any finite set~$S \subseteq \Sigma^*$ the equivalence relation~$\eqvr$ partitions the elements of~$S$ into a number of classes that is polynomially bounded in the size of the largest element of~$S$.
\end{itemize}
\end{definition}

\begin{definition}[{Cross-composition, \cite[Def. 3.3]{BodlaenderJK14}}]\label{definition:crosscomposition}
Let~$L\subseteq\Sigma^*$ be a language, let~$\eqvr$ be a polynomial equivalence relation on~$\Sigma^*$, let~$\Q\subseteq\Sigma^*\times\N$ be a parameterized problem, and let~$f \colon \N \to \N$ be a function. An \emph{\OR-cross-com\-position of~$L$ into~$\Q$} (with respect to \eqvr) \emph{of cost~$f(t)$} is an algorithm that, given~$t$ instances~$x_1, x_2, \ldots, x_t \in \Sigma^*$ of~$L$ belonging to the same equivalence class of~$\eqvr$, takes time polynomial in~$\sum _{i=1}^t |x_i|$ and outputs an instance~$(y,k) \in \Sigma^* \times \mathbb{N}$ such that:
\begin{itemize}
\item The parameter~$k$ is bounded by $\Oh(f(t)\cdot(\max_i|x_i|)^c)$, where~$c$ is some constant independent of~$t$, and
\item instance $(y,k) \in \Q$ if and only if there is an~$i \in [t]$ such that~$x_i \in L$.\label{property:OR}
\end{itemize}
\end{definition}

\begin{theorem}[{\cite[Theorem 6]{BodlaenderJK14}}] \label{thm:cross_composition_LB}
Let~$L\subseteq\Sigma^*$ be a language, let~$\Q\subseteq\Sigma^*\times\N$ be a parameterized problem, and let~$d,\varepsilon$ be positive reals. If~$L$ is NP-hard under Karp reductions, has an \OR-cross-composition into~$\Q$ with cost~$f(t)=t^{1/d+o(1)}$, where~$t$ denotes the number of instances, and~$\Q$ has a polynomial (generalized) kernelization with size bound~$\Oh(k^{d-\varepsilon})$, then \containment.
\end{theorem}

We will refer to an \OR-cross-composition of cost~$f(t) = \sqrt{t} \log (t)$ as a \emph{degree-$2$ cross-composition}. By Theorem~\ref{thm:cross_composition_LB}, a degree-$2$ cross-composition can be used to rule out generalized kernels of size~$\Oh(k^{2 - \varepsilon})$.

\section{Kernel for \texorpdfstring{$\mathbf{H}$}{H}-Coloring parameterized by Twin-Cover}
\label{section:kernel}

In this section, we give a kernel for \pname{$H$-Coloring} parameterized by the size of a twin-cover. We start by showing how to partition the graph into vertex sets that are twins in Section \ref{subsection:twin_decomposition}. We introduce some of the polynomial equalities that we use and their properties in Section \ref{subsection:polynomial_equalities}, and use them in Section \ref{subsection:list_of_equalities} to define the set of equalities that is constructed for a given input graph. In Section \ref{subsection:rrules} we define the three reduction rules our kernel will use and prove that they are safe. Finally, in Section \ref{subsec:kernel} we give the kernel.

\subsection{Twin Decomposition}\label{subsection:twin_decomposition}
Computing a minimum \pname{Twin-Cover} is NP-hard, since \pname{Vertex Cover} is NP-hard on graphs where no two vertices are twins.  We will therefore construct the kernel for \pname{$H$-Coloring} without knowing a twin-cover of the input graph. In order to do this, we decompose the graph into vertex sets consisting of twins. Recall that throughout the paper, twins are vertices with the same \emph{closed} neighborhood.

\begin{definition}[(Partial) twin decomposition]
A \emph{partial twin decomposition} of a graph $G$ is a partition~$\Pi = \{P_1, \ldots, P_m\}$ of $V(G)$, such that any two vertices in the same partite set are twins. Partition $\Pi$ is a \emph{twin decomposition} if furthermore any two vertices in different partite sets are not twins.
\end{definition}
To be able to use the twin decomposition for the kernelization procedure, we show how it can efficiently be computed. 

\begin{lemma}\label{lem:twin_decomposition:find}
A twin decomposition can be computed in~$\Oh(n+m)$ time.
\end{lemma}
\begin{proof}
This is for example stated in \cite[Exercise 2.17]{EfficientGraphRepresentations} for the case of finding false twins, which are vertices such that $N_G(u) = N_G(v)$. Finding (true) twins is similar. An example solution uses the adjacency-list representation, and adds each vertex to its own adjacency list. Then we efficiently sort the vertices based on their adjacency lists and use this to find duplicates.
\end{proof}
The next lemma shows how the twin decomposition and a minimal twin-cover may intersect.

\begin{lemma}\label{lem:twin_decomposition:twin_cover}
Let $G$ be a graph with twin decomposition $\Pi$ and a minimal twin-cover $S$. Then for any partite set $P \in \Pi$ it holds that either $P \subseteq S$ or $P \cap S = \emptyset$.
\end{lemma}
\begin{proof}
Let $P \in \Pi$.
Suppose $P \cap S \neq \emptyset$ and $P\setminus S\neq \emptyset$.  Let $u \in S \cap P$ and $v \in P \setminus S$. We show that $S \setminus \{u\}$ is a twin-cover of $G$, which contradicts the assumption that $S$ is minimal.

Let $\{u,w\}$ for $w \neq v$ be any edge in $G$. Since $u$ and $v$ are twins, it follows that $\{v,w\} \in E(G)$. Thereby, either $w \in S$ and thus edge $\{u,w\}$ is covered by $w$, or $w$ and $v$ are twins. In this case, by transitivity of being twins $u$ and $w$ are also twins.
This proves that $S \setminus \{u\}$ is indeed a twin-cover of $G$, which is a contradiction.
\end{proof}

\subsection{Modeling constraints as polynomial equalities}\label{subsection:polynomial_equalities}
As explained in the introduction, the kernelization is based on a connection to constraint satisfaction problems. To find the kernel, we represent the constraints that a vertex set puts on the coloring of its neighborhood, as polynomial equalities. We then use this representation to find redundant  vertices and edges in the graph. To use this idea, we need some additional lemmas and definitions. Recall that a monomial of degree~$d$ is the product of~$d$ variables, with the unique monomial of degree zero being the constant~$1$. For example,~$x_1 \cdot x_3 \cdot x_3$ is a monomial of degree three. A monomial is \emph{multilinear} if each variable occurs at most once.

\begin{lemma}[{cf.~\cite[Claim 4]{JansenP2016Optimal}}] \label{lem:num_monomials}
There are at most $n^d+1$ multilinear monomials of degree at most~$d$ over a set of~$n$ variables.
\end{lemma}
\begin{proof}
The number of multilinear monomials over $n$ variables of degree at most $d$ is equal to $\sum_{i=0}^d \binom{n}{i}$. We will show that $\sum_{i=1}^d \binom{n}{i} \leq n^d$. The left side counts all non-empty subsets of $[n]$ of size at most $d$. Each of these can be mapped to a distinct $d$-tuple containing numbers from $[n]$, by repeating an arbitrary element. Since there are $n^d$ possible $d$-tuples, the claim follows.
\end{proof}

\begin{definition}[span]
Let $X$ be a set of vectors in~$\{0,1\}^d$ for some~$d \in \mathbb{N}$. Define $\vectspan(X)$ as the set of all vectors $\vect{y} \in \{0,1\}^d$ for which there exist $\vect{x_1},\ldots,\vect{x_\ell}\in X$ such that $\vect{y} \equiv_2 \sum_{i=1}^\ell \vect{x}_i$, i.e.,~$\vect{y}$ is a linear combination of $\vect{x_1},\ldots,\vect{x_\ell}$ when vectors are added component-wise, over the integers modulo $2$.

Let $p(x_1, \ldots, x_n)$ be a multivariate polynomial in (a subset of) the variables $x_1, \ldots, x_n$, evaluated over the integers modulo~$2$, of degree at most~$d$ for some fixed~$d$. Hence~$p$ is a weighted sum of monomials of degree at most~$d$ over~$x_1, \ldots, x_n$. For some fixed ordering of the monomials of degree~$d$ over~$x_1, \ldots, x_n$, let $\somevector(p)$ denote the vector containing the coefficients of the corresponding monomials in~$p$.

Let $P$ be a set of multivariate polynomials in (subsets of) the variables $x_1, \ldots, x_n$. We use~$\vectspan(P)$ to denote $\vectspan{(\{ \somevector(p) \mid p \in P\})}$, and we use $p \in \vectspan(P)$ to denote that $\vectspan(P)$ contains $\somevector(p)$.
\end{definition}
The following lemma follows from the definition above.
\begin{lemma}\label{lem:span}
Let $P$ be a set of polynomials of degree at most~$d$ over variable set $\vect{y}$, and let $q$ be a polynomial of degree at most~$d$ over $\vect{y}$. If $q \in \vectspan(P)$, then any assignment to $\vect{y}$ that satisfies $p(\vect{y}) \equiv_2 0$ for all $p \in P$, satisfies $q(\vect{y}) \equiv_2 0$.
\end{lemma}
\begin{proof}
Choose $\alpha_p\in \{0,1\}$ for all $p \in P$  such that $\somevector(q) \equiv_2 \sum_{p \in P}\alpha_p \somevector(p)$. Consider an assignment to the variables~$\vect{y}$ with~$p(\vect{y}) \equiv_2 0$ for all~$p \in P$. Let $\vect{y'}$ be the vector containing the evaluation of the monomials of degree at most $d$ over~$\vect{y}$, for the values assigned to~$\vect{y}$. List them in the same order in which the coefficients for these monomials are listed in $\somevector(\cdot)$. Since a polynomial is a weighted sum of monomials, the value of a polynomial~$p$ of degree most~$d$ in~$\vect{y}$ for the assigned values, equals the inner product of~$\somevector(p)$ and~$\vect{y'}$. So:
\[
q(\vect{y}) = \somevector(q)\cdot \vect{y'} \equiv_2 \sum_{p \in P} \alpha_p\somevector(p)\cdot \vect{y'} \equiv_2 \sum_{p \in P} \alpha_p \cdot p(\vect{y}) \equiv_2 0.
\qedhere\]
\end{proof}

To utilize polynomials over boolean variables to represent solutions of graph $H$-coloring problems, we represent the color of a vertex $v$ in a graph $G$ by $|V(H)|$ boolean variables, indicating whether $v$ has the corresponding color. We now define a partial choice assignment, which reflect that any vertex receives at most one color.

\begin{definition}[Partial choice assignment]
Let $y_{i,k}\in\{0,1\}$ for $i \in [n]$, $k \in [q]$ be a set of boolean variables and let $\vect{y}$ be the vector containing all these variables. We say $\vect{y}$ is given a \emph{partial choice assignment} if for all $i \in [n]$:
\[\sum_{k = 1}^q y_{i,k} \leq 1.\]
\end{definition}

Note that a partial choice assignment sets at most $n$ variables to \true. By this definition, a partial choice assignment can be seen as a partial coloring in the following way: $y_{i,k} = 1$ means vertex $i$ has color $k$. Note that the coloring of some vertices may remain undefined.

The following lemma gives a polynomial that can be used to express the constraint that out of exactly~$q$ neighbors of a given vertex~$u$, there are at least two that have the same color. By combining multiple such constraints, we can ensure that at most $q-1$ different colors are used in the neighborhood of vertex $u$, leaving one color free for~$u$ itself in the $q$-coloring problem. When evaluating the polynomial for $\vect{y}$ that is given a partial choice assignment, the polynomial has the following two essential properties. (1) It equals $1$ modulo $2$ when the~$q$ vertices all receive a distinct color, and (2) it equals $0$ modulo $2$ whenever two vertices have the same color, or when two vertices have no color defined.

\begin{lemma}\label{lem:polynomial}
Let $q > 0$ be an integer and let $y_{i,k}$ for $i \in [q], k \in [q]$ be  boolean variables.  Then there exists a polynomial $p$ of degree $q-1$ over the integers modulo 2, such that whenever the variables in~$\vect{y}$ are given a partial choice assignment, it holds that $p(\vect{y}) \equiv_2 1$ if and only if
\begin{itemize}
\item there exist no $i,j,k \in [q]$ with $i \neq j$ such that $y_{i,k} = y_{j,k} = 1$, and
\item for all $k \in [q-1]$ there exists $i \in [q]$ such that $y_{i,k} = 1$.
\end{itemize}
\end{lemma}
\noindent Before proving Lemma \ref{lem:polynomial}, we give the polynomial $p$ corresponding to $q = 3$ as an example.
\begin{align*}
p(\vect{y}) :=&  \sum_{i_1\neq i_2 \in [3]} \prod_{k=1}^2   y_{i_k,k}\\
=&\ y_{1,1} \cdot y_{2,2} +y_{1,1} \cdot y_{3,2} +y_{2,1} \cdot y_{1,2}  +y_{2,1} \cdot y_{3,2}+y_{3,1} \cdot y_{1,2} +y_{3,1} \cdot y_{2,2}.
\end{align*}
One may verify for this example that letting $y_{1,1} = y_{2,2} = y_{3,3} = 1$ and all other variables zero, gives $p(\vect{y}) = 1 \equiv_2 1$, as desired. Setting $y_{1,1} = y_{2,2} = y_{3,2} = 1$ and all other variables to zero, gives $p(\vect{y}) = 2 \equiv_2 0$, which explains why the modulus is used. Furthermore, letting $y_{1,1} = y_{2,2} = 1$ and all other variables be zero, also results in $p(\vect{y}) \equiv_2 1$. 
We now proceed with the general construction.
\begin{proof}[Proof of Lemma \ref{lem:polynomial}.]
Define the multivariate polynomial $p$ as
\[ \label{eq:coloring_polynomial}
p(\vect{y}) := \sum_{\substack{i_1,\ldots,i_{q-1} \in [q]\\ \text{distinct}}}  \prod_{k=1}^{q-1} y_{i_k,k}. 
\]

We prove that~$p$ has the desired properties. It is easy to see the degree of $p$ is $q-1$. It remains to prove the claim on the values of~$p(\vect{y})$ for partial choice assignments. So let~$\vect{y}$ be given a partial choice assignment, and for each~$i \in [q]$ let $x_i := k$ exactly when $y_{i,k} = 1$. Let $x_i:=0$ when there is no such $y_{i,k}$.

We now show that $p(\vect{y}) \equiv_2 1$ if there exist no $i,j,k \in [q]$ with $i \neq j$ such that $y_{i,k} = y_{j,k} = 1$, and for all $k \in [q-1]$ there exists $i \in [q]$ such that $y_{i,k} = 1$. In terms of the values for $x_i$, this implies that they are all distinct, and that $[q-1] \subseteq \{x_1,\ldots,x_q\}$. Thereby,  we only have to consider the following two cases. Either $\{x_1,\ldots,x_q\} = [q]$, or $\{x_1,\ldots,x_q\} = [q-1]\cup \{0\}$.

Suppose that we are in one of the two situations above. For $k \in [q-1]$, let $j_k$ be the unique index such that $x_{j_k} = k$, implying that~$y_{j_k,k} = 1$. Note that this is well defined, since all values from $[q-1]$ are used exactly once. Then, $\prod_{k=1}^{q-1} y_{j_k,k} = 1$. For any other choice of distinct indices $i_1,\ldots, i_{q-1} \in [q]$, there exists $m \in [q-1]$ such that $i_m \neq j_m$. This implies that $y_{i_m,m} = 0$ and thereby $\prod_{k=1}^{q-1} y_{i_k,k} = 0$. Thus, $p(\vect{y}) = 1 \equiv_2 1$.

Now  suppose there exist $i,j\in[q]$, such that $x_i = x_j\neq 0$, or there exists $k \in [q-1]$ such that  $y_{i,k} = 0$ for all~$i \in [q]$. We show that $p(\vect{y}) \equiv_2 0$ by a case distinction.
\begin{itemize}
\item There exists $k \in [q-1]$ such that $\sum_{i=1}^q y_{i,k} = 0$, or equivalently there is no $i \in [q]$ such that $x_i = k$. Thereby, for any choice of $i_1,\ldots, i_{q-1} \in [q]$, we have that $ \prod_{\ell=1}^{q-1} y_{i_\ell,\ell} = 0$, since $ y_{i_k,k} = 0$. Thus, $p(\vect{y}) \equiv_2 0$.
\item There exists no $k \in [q-1]$ such that $\sum_{i=1}^q y_{i,k} = 0$. Thereby, for each $k \in [q-1]$ there exists $i \in [q]$ such that $x_i = k$. It follows from our earlier assumption that there must exist $i,j,k \in [q]$ with $i \neq j$ such that $x_i = x_j = k$, which implies~$k < q$.
    For all $c\in [q-1]$ with $c\neq k$, let $i_c$ be the unique index such that $x_{i_c} = c$ and thus $y_{{i_{c}},c} = 1$. Then
    \[y_{i,k} \cdot \prod_{\substack{c = 1\\c \neq k}}^{q-1} y_{{i_{c}},c} = y_{j,k} \cdot\prod_{\substack{c = 1\\c \neq k}}^{q-1} y_{{i_{c}},c} = 1.\]
     However, $\prod_{c = 1}^{q-1} y_{{i_{c}},c} =  0$ for any other choice of $i_1,\ldots,i_{q-1}$. Thereby, $p(\vect{y}) = 2 \equiv_2 0$.\qedhere
\end{itemize}
\end{proof}


\subsection{Construction of polynomial equalities}\label{subsection:list_of_equalities}
We continue to define the polynomial equalities that will be constructed for a subset $P$ of the vertices of $G$. These are necessary constraints on the coloring of $N_G(P)$, such that $P$ can be properly $H$-colored. In the construction, $P$ will be a partite set of the twin decomposition of $G$, and hence a clique.

Let $G$ be a graph with $P \subseteq V(G)$. We create variables $c_{v,i}$ for each $v \in V(G)$ and $i \in V(H)$, denoting whether $v$ has color $i$. Let $\vect{C}$ contain all constructed variables. Let $L(P,G)$  be the set of polynomial equalities produced by the following procedure, which results in two types of constraints. The first will ensure that the neighborhood of $P$ does not use too many colors, such that there are at least $|P|$ remaining colors to color (the clique) $P$. The second will ensure that the coloring of the neighborhood of $P$ can be extended to also color $P$.

For each set $S \subseteq N_G(P)$ with $|S| =  \Delta(H) + 1$ and each set $X \subseteq V(H)$ with $|X| = |S|$, use Lemma \ref{lem:polynomial} to find a polynomial $p_{P,S,X}$ such that $p_{P,S,X}(\vect{C}) \equiv_2 1$ if and only if the following two statements hold:

\begin{enumerate}
	\item there exist no $u\neq v \in S$, $k \in X$ such that $c_{u,k} = c_{v,k} = 1$, and
	\item let $X = \{x_1,\ldots,x_{|S|}\}$ then for all $k \in [|S|-1]$ there exists $u \in S$ such that $c_{u,x_k} = 1$.
\end{enumerate}
Add the  following constraint to $L(P,G)$:
\[p_{P,S,X}(\vect{C}) \equiv_2 0.\]

For the second type of constraints, consider each set $S \subseteq N_G(P)$ of $k \in [\Delta(H)]$ elements, and each sequence $x_1,\ldots,x_k \in V(H)$ (of not necessarily distinct elements). Let $S = \{s_1,\ldots,s_k\}$.
Let $Y:= \bigcap_{i=1}^{k} N_H(x_i)$ be the common neighborhood of $X$ in $H$.
If $H[Y]$ does not contain a clique of size at least $|P|$ (i.e. if $\omega(H[Y]) < |P|$), add the following polynomial equality to $L(P,G)$:
\[q_{P,S,X}(\vect{C}) := \prod_{i =1}^{k} c_{s_i,x_i}\equiv_2 0.\]
The computation of $\omega(H[Y])$ for some $Y\subseteq V(H)$ can be done in constant time, since $H$ is considered constant.
This concludes the construction of $L(P,G)$. Note that the constraints~$L(P,G)$ are defined solely in terms of the variables that describe the coloring of the \emph{open neighborhood} of~$P$.

Next, we define a complete list of equations for $G$.
\begin{definition}[{$\mathbf{L_\Pi(G)}$}]
Let $G$ be a graph and let $\Pi$ be a partition of its vertex set. Let $L_\Pi(G)$ be defined as follows.
\[L_{\Pi}(G) := \bigcup_{P \in \Pi} L(P,G).\]
\end{definition}
Since the polynomials for the first type of constraints have degree at most~$|S| - 1= \Delta(H)$ by Lemma~\ref{lem:polynomial}, while the polynomials for the second type of constraints are the product of~$k \leq \Delta(H)$ variables, we observe the following.

\begin{observation}\label{obs:deg_L}
Let $G$ be a graph with $\Pi$ a partition of its vertex set. The polynomials in $L_{\Pi}(G)$ have degree at most $\Delta(H)$.
\end{observation}

We now prove two useful properties for this choice of constraints.

\begin{lemma}\label{lem:constraints_satisfied_by_f}
Let $G$ be a graph with some partial twin decomposition $\Pi$. Let $f\colon V(G)\rightarrow V(H)$ be some coloring of the vertices of $G$. If  $f$ is a proper $H$-coloring of $G$, then the boolean assignment to $\vect{C}:=\{c_{v,i } \mid v \in V(G), i \in V(H)\}$ given by $c_{v,i} = 1 \Leftrightarrow f(v) = i$ satisfies all constraints in  $L_{\Pi}(G)$.
\end{lemma}
\begin{proof}
Let $f$ be given and the value of any $c_{v,i}\in \vect{C}$ be defined by $c_{v,i} = 1 \Leftrightarrow f(v) = i$.
We show that this assignment satisfies all constraints in $L_{\Pi}(G)$, by showing that it satisfies both types of constraints in $L(P,G)$ for all $P \in \Pi$. Consider some~$P \in \Pi$. Since it consists of twins, it is a clique in~$G$. As~$H$ has no self-loops, the vertices in $P$ all receive distinct colors by Observation~\ref{obs:clique:hcoloring}, and the colors used on $P$ form a clique in $H$. The fact that $P$ consists of twins also implies that $\{u,v\} \in E(G)$ for all $u \in P, v \in N_G(P)$. Thereby, any color used in $P$ is not used in the coloring of $N_G(P)$.

Consider a constraint $p_{P,S,X}(\vect{C}) \equiv_2 0 \in L(P,G)$ for $S \subseteq N_G(P)$ of size~$|\Delta(H)|+1$ and $X \subseteq V(H)$ of the same size. By Observation~\ref{obs:numcolors:delta}, the vertices in $S$ use at most $\Delta(H) = |S|-1$ colors. Thereby, some color in $X$ is used twice, or at least two colors in $X$ are unused. It follows from Lemma \ref{lem:polynomial} that $p_{P,S,X}(\vect{C}) \equiv_2 0$ as required.

Consider a constraint $q_{P,S,X}(\vect{C}) \equiv_2 0 \in L(P,G)$ for $S \subseteq N(P)$ and $X = x_1,\ldots,x_{|S|} \in V(H)$. Suppose this constraint is not satisfied. Then the coloring of $S$ is  given by $X$ and furthermore, $H[Y]$ where $Y := \bigcap_{i=1}^{|S|} N_H(x_i)$ does not contain a clique on $|P|$ vertices. But a proper $H$-coloring (for~$H$ without self-loops) colors any clique in~$G$ with an equally-sized clique in~$H$, and the colors used on the clique~$P$ must be adjacent in~$H$ to all the colors used on the neighbors~$S$ of~$P$ in~$G$. Since~$H[Y]$ contains no clique on~$|P|$ vertices, $f$ cannot be a proper $H$-coloring of $G$. It follows that for proper $H$-colorings, all constraints are satisfied.
\end{proof}

Let $S\subseteq V(G)$ and let $f \colon S \rightarrow V(H)$ be  a proper $H$-coloring of $G[S]$. We say that $f$ can be \emph{extended} to properly color $G$, if there exists $f' \colon V(G) \rightarrow V(H)$ such that $f'$ is a proper $H$-coloring of $G$ and furthermore $f'(v) = f(v)$ for all $v \in S$. The next lemma shows that an $H$-coloring of a part of the graph can be extended to color the entire graph, if it satisfies certain constraints.

\begin{lemma}\label{lem:can_extend_to_v}
Let $G$ be a graph with $P' \subseteq V(G)$. 
Let $f\colon V(G)\setminus P' \rightarrow V(H)$ be a proper $H$-coloring of $G-P'$, such that  the boolean assignment to $\vect{C}:=\{c_{v,i } \mid v \in V(G), i \in V(H)\}$ given by $c_{v,i} = 1 \Leftrightarrow v \notin P' \wedge f(v) = i$ satisfies all constraints in $L(P',G)$. Then $f$ can be extended to properly color $G$.
\end{lemma}
\begin{proof}
Let $f$ be given and $\vect{C}$ be defined by $c_{v,i} = 1 \Leftrightarrow v \notin P' \wedge f(v) = i$.
We start by showing that $N_G(P')$ uses at most $\Delta(H)$ different colors. Suppose not, then there is a set $S \subseteq N_G(P')$ of size $\Delta(H)+1$ using $\Delta(H)+1$ distinct colors. Let $X$ be the set of colors used in $S$.
It follows from Lemma \ref{lem:polynomial}, that $p_{P',S,X}(\vect{C})\equiv_2 1$. 
By definition, $L(P',G)$ contains the equation $p_{P',S,X}(\vect{C})\equiv_2 0$. This contradicts the assumption that all constraints in $L(P',G)$ are satisfied.

Let $X = \{x_1,\ldots,x_\ell\}$ be the set of colors used by $N_G(P')$. We know that $|X| \leq \Delta(H)$. Let $S = \{s_1.\ldots, s_{\ell}\}$ be a subset of $N_G(P')$ such that $f(s_i) = x_i$ for all $i \in [\ell]$. Thereby, $c_{s_i,x_i} = 1$ for all $i \in [\ell]$. By this definition, $q_{P',S,X}(\vect{C}) \equiv_2 1$ and thus $q_{P',S,X}(\vect{C}) \equiv_2 0 \notin L(P',G)$.
 Thereby, $G[Y]$ contains a clique $K$ of size $|P'|$, where $Y := \bigcap_{i=1}^{|X|} N_H(x_i)$. To extend $f$ to color $P'$, assign each vertex in $P'$ a distinct color from $K$.

 It remains to verify that we have given a proper $H$-coloring. Any edge between two vertices in $V(G)\setminus P'$ remains properly colored. Any edge in $P'$ is properly colored, because its endpoints have a different color and $K$ is a clique in $H$. Any edge between $P'$ and $V(G)\setminus P'$ is properly colored, because all vertices in $K$ are a common neighbor of the vertices in $X$, and $K \cap X = \emptyset$.
\end{proof}

\subsection{Reduction rules}\label{subsection:rrules}
We now present the three reduction rules that will be used to obtain the kernel, and prove that they are safe. The first checks whether the graph is trivially not $H$-colorable, the second removes sets of edges from the graph, and the third removes sets of vertices from the graph.

\begin{rrule}\label{rrule:trivial}
Let $G$ be a graph with twin decomposition $\Pi$. If there exists $P \in \Pi$ with $|P| > \omega(H)$,
return a trivial no-instance.
\end{rrule}

It is easy to see that Reduction rule \ref{rrule:trivial} preserves the answer to the problem, since $G$ cannot have a proper $H$-coloring by Observation~\ref{obs:clique:hcoloring}.

\begin{rrule}\label{rrule:remove-edge}
Let $G$ be a graph with twin decomposition $\Pi$. Let $P' \neq P'' \in \Pi$ such that $E_G(P',P'')\neq \emptyset$. If  $L(P',G) \subseteq \vectspan(L_{\Pi}(G \setminus E_G(P',P'')))$,  remove all edges in $E_G(P',P'')$ from graph $G$.
\end{rrule}

Reduction rule~\ref{rrule:remove-edge} is the key rule for our kernelization. It simplifies the graph by removing all edges between two distinct sets of twins~$P'$ and~$P''$, if the constraints~$L(P',G)$ are linear combinations of the constraints generated by the remaining graph~$G \setminus E_G(P',P'')$. The following lemma proves that the reduction rule is safe.

\begin{lemma}
If $G'$ is obtained from $G$ by applying Reduction rule \ref{rrule:remove-edge}, then $G$ is $H$-colorable if and only if $G'$ is $H$-colorable.
\end{lemma}
\begin{proof}
Let $G'$ be $G \setminus  E_G(P',P'')$.
Clearly, if $G$ is $H$-colorable, then $G'$ is also $H$-colorable, since $G'$ is a subgraph of $G$.

In the other direction, let $f'$ be a proper $H$-coloring of $G'$. It follows from Lemma \ref{lem:constraints_satisfied_by_f} and the fact that $\Pi$ is a partial twin decomposition of $G\setminus E_G(P',P'')$ that the derived setting of the boolean variables $\vect{C}$ satisfies the constraints in $L_{\Pi}(G \setminus E_G(P',P''))$. Since $L(P',G)\subseteq \vectspan(L_{\Pi}(G \setminus E_G(P',P'')))$ it follows from Lemma \ref{lem:span} that this setting of $\vect{C}$ also satisfies all constraints in $L(P',G)$. Let $f$ be defined as $f'$ restricted to the vertices in $G'-P'$. Note that $G'-P'$ equals $G-P'$ by definition. It is easy to see that $f$ is a proper $H$-coloring of $G-P'$ since $G-P'$ is a subgraph of $G'$ and $f'$ is a proper $H$-coloring of $G'$. Furthermore, $f$ satisfies the constraints in $L(P',G)$ since it colors the relevant vertices the same as $f'$. It now follows from Lemma \ref{lem:can_extend_to_v} that we can extend $f$ to color all vertices in $G$. Thereby, $G$ has a proper $H$-coloring.
\end{proof}

The final rule effectively removes isolated cliques from the graph, when~$H$ has a sufficiently large clique to allow them to be colored properly.

\begin{rrule}\label{rrule:remove-vertex}
Let $G$ be a graph with twin decomposition $\Pi$. If there exists $P' \in \Pi$ with $N_G(P') = \emptyset$ and~$|P'| \leq \omega(H)$, then remove $P'$ from $G$.
\end{rrule}

\begin{lemma}
If $G'$ is obtained from $G$ by applying Reduction rule \ref{rrule:remove-vertex}, then $G$ is $H$-colorable if and only if $G'$ is $H$-colorable.
\end{lemma}
\begin{proof}
Let $P'$ be such that $G' = G - P'$.
Clearly, if $G$ is $H$-colorable, $G'$ remains $H$-colorable. In the other direction, suppose $G'$ is $H$-colorable. We show how to extend this coloring to $G$. Since we assumed that $|P'| \leq \omega(H)$, $H$ has a clique $X$ of size $|P'|$. Use the colors of $X$ to assign a different color to each vertex in $P'$. Since $N_G(P')=\emptyset$, this gives a proper $H$-coloring of $G$.
\end{proof}

\begin{lemma}\label{claim:decrease_twin_cover}
Applying Reduction rule \ref{rrule:remove-edge} or \ref{rrule:remove-vertex} to a graph $G$ does not increase the size of a minimum twin-cover of $G$.
\end{lemma}
\begin{proof}
Let $P \subseteq V(G)$. It is easy to see that if $S$ is a twin-cover of $G$, then it is also a twin-cover of $G-P$. Thereby, the statement holds for Reduction rule \ref{rrule:remove-vertex}.

Let $\Pi$ be the twin decomposition of~$G$ and let $P' \neq P'' \in \Pi$. Let $F:= E_G(P',P'')$ be the set of edges between $P'$ and $P''$. Let $S$ be a twin-cover of $G$. We show that $S$ is a twin-cover of $G \setminus F$. Clearly, any edges that were previously covered, are still covered. We show that all vertices that were twins in $G$, are also twins in $G \setminus F$ to conclude the proof. Let $u,v$ be twins in $G$, and let $P\in \Pi$ such that $u,v \in P$. If $P \neq P'$ and $P\neq P''$, it is obvious that $u$ and $v$ remain twins in $G \setminus F$. Suppose $u,v$ lie in $P'$ or in $P''$; without loss of generality, suppose $u,v\in P'$. Note that the edge $\{u,v\}$ belongs to~$E(G) \setminus F$. Then $N_{G \setminus F}[u] = N_G[u] \setminus P''$. Since $u$ and $v$ are twins in $G$, $N_G[u]\setminus P'' = N_G[v]\setminus P'' = N_{G \setminus F}[v]$. Thereby, $u$ and $v$ are twins in $G\setminus F$.  It follows that the lemma statement also holds for Reduction rule \ref{rrule:remove-edge}.
\end{proof}

\subsection{Analysis of the kernelization}\label{subsec:kernel}

\begin{theorem}\label{thm:kernel}
For any fixed non-bipartite graph $H$ (without self-loops), \pname{$H$-Coloring} parameterized by the size $k$ of a twin-cover has a kernel with $\Oh(k^{\Delta(H)})$ vertices and edges, which can be encoded in $\Oh(k^{\Delta(H)}\log{k})$ bits. Furthermore, the kernelized instance is a subgraph of the original input graph.
\end{theorem}
\begin{proof}

Let $G$ be a graph. Apply Reduction rules \ref{rrule:trivial}, \ref{rrule:remove-edge}, and \ref{rrule:remove-vertex} exhaustively. Let the resulting graph be $G'$. We show that this is a correct kernelization.

\begin{myclaim}\label{claim:kernel:running-time}
Reduction rules \ref{rrule:trivial}--\ref{rrule:remove-vertex} can exhaustively be applied in polynomial time.
\end{myclaim}
\begin{claimproof}
We can compute a twin decomposition of $G$ in linear time by Lemma~\ref{lem:twin_decomposition:find}. Computing $\omega(H)$ can be done in $\Oh(1)$ time for fixed $H$. Hence Reduction rule~\ref{rrule:trivial} can be applied in polynomial time.

The set $L_\Pi(G)$ contains at most $m := 2n\cdot n^{\Delta(H)+1}\cdot|V(H)|^{\Delta(H)+1}$ polynomial equalities (the number of ways to pick $S$, $X$ and $P$ as for the definition of $p_{P,S,X}$ and $q_{P,S,X}$), over $n \cdot |V(H)|$ variables. All polynomials we employ are multilinear. This can be verified directly from their construction, and explained by noting that squaring a number does not change it, when working modulo 2. By Lemma \ref{lem:num_monomials}, we therefore only have to consider $(n\cdot |V(H)|)^{\Delta(H)}+1$ coefficients for the polynomials. Constructing the required polynomial equalities can be done in polynomial time, for fixed~$H$. We can test if one vector lies in the span of a set of other vectors by comparing the ranks of matrices of dimensions at most $m \times ((n\cdot |V(H)|)^{\Delta(H)}+1)$. Thereby, Reduction rule \ref{rrule:remove-edge} can be applied in polynomial time. Reduction rule \ref{rrule:remove-vertex} can trivially be applied in polynomial time. Since $|\Pi| \leq |V(G)|$, checking for all $P\in \Pi$ whether any of the reduction rules can be applied takes polynomial time.

Each rule can be applied at most $|V(G)|^2$ times, as it always removes at least one edge or vertex. The claim follows.
\end{claimproof}

Let $G'$ be the result of applying  Reduction rules  \ref{rrule:trivial}, \ref{rrule:remove-edge}, and \ref{rrule:remove-vertex} exhaustively. We use the following claim to prove a bound on the size of $G'$.
\sloppy
\begin{myclaim}\label{claim:kernel:size}
The resulting graph $G'$ has $\Oh\left(|V(H)|^{\Delta(H)}\Delta(H)^2\cdot k^{\Delta(H)}\right)$  vertices and $\Oh\left(|V(H)|^{\Delta(H)}\Delta(H)^3\cdot k^{\Delta(H)}\right)$ edges.
\end{myclaim}
\fussy
\begin{claimproof}
When Reduction rule \ref{rrule:trivial} has been applied at any point, $G'$ trivially has constant size. Otherwise, since $G$ has a twin-cover of size $k$, it follows from Lemma \ref{claim:decrease_twin_cover} that $G'$ has a twin-cover of size at most $k$. Let $Y$ be a minimum twin-cover of $G'$, such that $|Y| \leq k$. Let $\Pi$ be the twin decomposition of $G'$. By Lemma \ref{lem:twin_decomposition:twin_cover}, every $P$ in $\Pi$ is either fully contained in $Y$, or disjoint from $Y$. Let $\Pi' := \{P \in \Pi \mid P \cap Y = \emptyset\}$. Define $L_{tc} := \bigcup_{P \in \Pi'} L(P, G')$, and note that $N_{G'}(P) \subseteq Y$ for all $P \in \Pi'$. This implies the polynomial equalities in~$L_{tc}$ only involve the variables controlling the coloring of~$Y$. By Observation \ref{obs:deg_L}, the polynomials in $L_{tc}$ have degree at most $\Delta(H)$ and they use at most $|V(H)|$ variables for each of the~$k$ vertices in $Y$. Define \[\alpha := (k\cdot |V(H)|)^{\Delta(H)} +1.\]

Let $L_{tc}' \subseteq L_{tc}$ be a basis of the vectors of $L_{tc}$, working modulo 2. Since all employed polynomials are multilinear, it follows that the vectors in~$L_{tc}$ only have nonzero coefficients for positions corresponding to multilinear monomials, of which there are at most~$\alpha$ by Lemma \ref{lem:num_monomials}. As the size of the basis~$L_{tc}'$ equals the rank of the matrix containing the (row)vectors~$L_{tc}$, which is upper-bounded by the number of columns that contain a nonzero entry, it follows that $|L_{tc}'| \leq \alpha$.

We define a set of meta-edges $\F \subseteq (\Pi' \times (\Pi\setminus \Pi'))$ based on the constraints in $L'_{tc}$. For each constraint
$Z$ in $L'_{tc}$, do the following.
\begin{itemize}
\item Suppose $Z = p_{P',S,X}(\vect{C})\equiv_2 0$ for some $P' \in \Pi'$, $S \subseteq N_G(P')$ and $X \subseteq V(H)$. Since~$P'$ is a partite set of twins that is disjoint from~$Y$, we have~$N_G(P') \subseteq Y$ since~$Y$ is a twin cover. So each~$v\in S$ belongs to a partite set~$P_v$ of twins with~$P_v \in \Pi \setminus \Pi'$. For each~$v \in S$, add $(P',P_v)$ to~$F$.
\item Otherwise, $Z = q_{P',S,X}(\vect{C})\equiv_2 0$ for some $P' \in \Pi'$, $S \subseteq N_G(P')$, and sequence $X = x_1,\ldots,x_k \in V(H)$. Similarly as above, for each $v \in S$ take $P_v \in \Pi\setminus \Pi'$ such that $v \in P_v$ and add $(P',P_v)$ to~$F$.
\end{itemize}
The above procedure adds at most $\Delta(H)+1$ meta-edges for each constraint in $L_{tc}'$. Thereby,
\begin{equation}\label{eq:bound_on_F}|F| \leq \alpha(\Delta(H) + 1).\end{equation}
We now argue that for any $(P',P'')\notin F$ with $P'\in \Pi'$ and $P'' \in \Pi\setminus \Pi'$, the following holds:
\begin{equation}\label{eq:Ltcprime}
L_{tc}' \subseteq L_{\Pi}(G'\setminus E_{G'}(P',P'')).
\end{equation}
To see this, consider a constraint in~$L_{tc}'$. It is of one of two possible types, and it was added to~$L_{tc} = \bigcup_{P \in \Pi'} L(P, G') \supseteq L'_{tc}$ because it satisfied the criteria described in Section~\ref{subsection:list_of_equalities}. Effectively, the constraint was created because some set~$P \in \Pi'$ contains a certain vertex set~$S$ of size at most~$\Delta(H) + 1$ in its open neighborhood in~$G'$. But by our choice of meta-edges~$F$, the set~$P$ still has~$S$ in its neighborhood in~$G' \setminus E_{G'}(P, P'')$, so that all constraints of~$L'_{tc}$ are also contained in~$L_{\Pi}(G' \setminus E_{G'}(P',P''))$.

Using this, we show that for all $P'\in\Pi'$ and $P''\in\Pi\setminus \Pi'$:
\begin{equation}\label{eq:correct_meta-edges}
E_{G'}(P',P'')\neq \emptyset \Rightarrow (P',P'') \in \F.
\end{equation}
Suppose there exist $P' \in \Pi',P''\in\Pi\setminus \Pi'$ such that $E_{G'}(P',P'')\neq \emptyset$ but $(P',P'') \notin \F$. It follows from Equation \ref{eq:Ltcprime} that $L_{tc}' \subseteq L_{\Pi}(G'\setminus E_{G'}(P',P''))$. Thereby,
\[\vectspan(L_{\Pi}(G'\setminus E_{G'}(P',P''))) \supseteq \vectspan(L_{tc}')\supseteq L_{tc}\supseteq L(P',G').\]
Thereby, Reduction rule \ref{rrule:remove-edge} could be applied to $G'$, which is a contradiction. It follows that $P'\in\Pi'$ and $P'' \in \Pi\setminus \Pi'$ can only be connected in $G'$, if there is a corresponding meta-edge in $\F$.
We can now use Equations \ref{eq:bound_on_F} and \ref{eq:correct_meta-edges} to bound the number of vertices and edges in $G'$.

First of all, for all $P' \in \Pi'$ there must exist some $P'' \in \Pi\setminus \Pi'$ such that $(P',P'')\in \F$, otherwise it follows from Equation \ref{eq:correct_meta-edges} that $N_{G'}(P') = \emptyset$ and $P'$ would have been removed by Reduction rule \ref{rrule:remove-vertex}. Thereby $|\Pi'| \leq |\F|$. Since $|P| \leq \omega(H) \leq \Delta(H)+1$ for all $P \in \Pi$ by Reduction rule~\ref{rrule:trivial}, the number of vertices of $G'$ can be bounded as follows.
\begin{align*}
 |V(G')| \leq |\F| \cdot (\Delta(H)+1) + |Y|  &\leq \left((k|V(H)|)^{\Delta(H)}+1\right)\cdot(\Delta(H)+1)^2 + k \\
 &= \Oh\left(|V(H)|^{\Delta(H)}\Delta(H)^2\cdot k^{\Delta(H)}\right).
\end{align*}

If edge $\{u,v\} \in G'$ with $u \in Y$ and $v \notin Y$, then there exist $(P',P'')\in\F$ such that $ u\in P'$, $v \in P''$. Since $|P| \leq \Delta(H)+1$ for any $P \in \Pi$, there are at most $|\F| \cdot (\Delta(H)+1)^2$ such edges. Furthermore, there are at most $\binom{|Y|}{2}\leq k^2$ edges between vertices in $Y$, and at most $|\F|\cdot (\Delta(H)+1)^2$ edges between vertices in $V(G) \setminus Y$. Thereby, the total number of edges can be bounded by:
\begin{align*}
|E(G')| &\leq |\F| \cdot (\Delta(H)+1)^2 + |Y|^2 + |\F| \cdot (\Delta(H)+1)^2\\
& \leq 2\alpha(\Delta(H)+1)^3 + k^2\\
&= 2 \left((k\cdot |V(H)|)^{\Delta(H)} +1\right)(\Delta(H)+1)^3 + k^2\\
\intertext{(note that $\Delta(H) \geq 2$ for non-bipartite $H$)}
&= \Oh\left(|V(H)|^{\Delta(H)}\Delta(H)^3\cdot k^{\Delta(H)}\right).
\end{align*}
This concludes the proof of Claim~\ref{claim:kernel:size}.
\end{claimproof}
It follows from the correctness of Reduction rules \ref{rrule:trivial}, \ref{rrule:remove-edge}, and \ref{rrule:remove-vertex} that $G'$ is $H$-colorable if and only if $G$ is $H$-colorable. It follows from Claims \ref{claim:kernel:running-time} and \ref{claim:kernel:size} that we have given a kernel for $H$-coloring with  $\Oh(k^{\Delta(H)})$ vertices and edges for constant-size $H$ that can be computed in polynomial time. By encoding the graph using adjacency lists, it can be encoded in $\Oh(k^{\Delta(H)} \cdot \log{k})$ bits.
\end{proof}

The following corollary shows that Theorem \ref{thm:kernel} generalizes the result obtained for \qcoloring parameterized by vertex cover in the extended abstract of this work.

\begin{corollary}\label{corollary:kernel:q-coloring}
For any constant $q \geq 3$, \pname{$q$-Coloring} parameterized by the size of a twin-cover has a kernel with $O(k^{q-1})$ vertices, which can be encoded in $O(k^{q-1}\log{k})$
 bits. Furthermore, the resulting instance is a subgraph of the original input graph.
 \end{corollary}
\begin{proof}
Since \pname{$q$-Coloring} is equivalent to \pname{$K_q$-Coloring}, and $\Delta(K_q) = q-1$ and $K_q$ has $q$ vertices, the result now follows directly from Theorem \ref{thm:kernel}. 
\end{proof}

\section{Sparsification lower bound for \texorpdfstring{$3$}{3}-Coloring}

In this section we provide a sparsification lower bound for \pname{$3$-Coloring}. We show that \ThreeColoring does not have a (generalized) kernel of size $\Oh(n^{2-\varepsilon})$, unless \containment. This will also provide a kernel lower bound for \ThreeColoring parameterized by the size of a twin-cover, that matches the upper bound given in the previous section up to $k^{o(1)}$ factors.

For ease of presentation, we will prove the lower bound by giving a degree-$2$ cross-composition from a tailor-made problem to \pname{$3$-List Coloring}. The input to \pname{$3$-List Coloring} is a graph $G$ together with a function $L$ that assigns to each vertex~$v$ a list~$L(v) \subseteq \{1,2,3\}$. The problem asks whether there exists a proper coloring of $G$, such that each vertex is assigned a color from its list.
Before presenting the cross-composition, we introduce an important gadget that will be used. It was constructed by Jaffke and Jansen \cite{JaffkeJ17Fine}. The gadget, which we will call a \newgadget, will be used to forbid one specific coloring of a given vertex set. The following Lemma is a rephrased version of Lemma 15 in \cite{JaffkeJ17Fine}.

\sloppy
\begin{lemma}
\label{lem:blocking_gadget}
There is a polynomial-time algorithm that, given $\vect{c} = (c_1,\ldots,c_m) \in [3]^m$, outputs a \pname{$3$-List-Coloring} instance with $\Oh(m)$ vertices called $\newgadget(\vect{c})$ that contains distinguished vertices $(\pi_1,\ldots,\pi_m)$. A coloring $f\colon \{\pi_i\mid i\in [m]\} \rightarrow [3]$ can be extended to a proper list coloring of $\newgadget(\vect{c})$ if and only if $f(\pi_i) = c_i$ for some $i \in [m]$.
\end{lemma}
\fussy
The \newgadget can be used to forbid one specific coloring given by the tuple $\vect{c}$ of a set of vertices $v_1,\ldots,v_m$, by adding a $\newgadget(\vect{c})$ and connecting $\pi_i$ to $v_i$ for all $i \in [m]$. If the color of $v_i$ is $c_i$ for all $i$, then the inserted edges prevent all~$\pi_i$ to receive the corresponding color~$c_i$, and by Lemma \ref{lem:blocking_gadget} the coloring cannot be extended to the gadget. If however the color of $v_i$ differs from $c_i$ for some $i$, the gadget can be properly colored.

Having presented the gadget we use in our construction, we define the source problem for the cross-composition. This problem was also used as the starting problem for a cross-composition in our earlier sparsification lower bound for \pname{$4$-Coloring}~\cite{JansenP2016Sparsification}.

\defproblem{\pname{$2$-$3$-Coloring with Triangle Split  Decomposition}\cite{JansenP2016Sparsification}}
{A graph $G$ with a partition of its vertex set into $U \cup V$ such that $G[U]$ is an edgeless graph and $G[V]$ is a disjoint union of triangles.}
{Is there a proper $3$-coloring $c:V(G) \to \{1,2,3\}$ of $G$, such that $c(u)~\in~\{1,2\}$ for all $u \in U$? We will refer to such a coloring as a \emph{2-3-coloring} of the graph $G$, since two colors are used to color $U$, and three to color $V$.}

\begin{lemma}[{\cite[Lemma 3]{JansenP2016Optimal}}]
\label{lem:NP_complete_coloring}
 \pname{$2$-$3$-Coloring with Triangle Split Decomposition}  is NP-complete.
\end{lemma}

To establish a quadratic lower bound on the size of generalized kernels, it suffices to give a degree-$2$ cross-composition from this special coloring problem into \pname{$3$-Coloring}. Effectively, we have to show that for any~$t$, one can efficiently embed a series of~$t$ size-$n$ instances indexed as~$X_{i,j}$ for~$i,j \in [\sqrt{t}]$, into a single \pname{$3$-Coloring} instance with~$\Oh(\sqrt{t} \cdot n^{\Oh(1)})$ vertices that acts as the logical \textsc{or} of the inputs. To achieve this composition, a common strategy is to construct vertex sets $S_i$ and $T_i$ of size~$n^{\Oh(1)}$ for $i \in [\sqrt{t}]$, such that the graph induced by $S_i \cup T_j$ encodes input~$X_{i,j}$. The fact that the inputs can be partitioned into an independent set and a collection of triangles facilitates this embedding; we represent the independent set within sets~$S_i$ and the triangles in sets~$T_i$. To embed~$t$ inputs into a graph on~$\Oh(\sqrt{t} \cdot n^{\Oh(1)})$ vertices, each vertex will have incident edges corresponding to many different input instances. The main issue when trying to find a cross-composition into \pname{$3$-Coloring}, is to ensure that when there is one $2$-$3$-colorable input graph, the entire graph becomes $3$-colorable. This is difficult, since the neighbors that a vertex in $S_i$ has among the many different sets~$T_j$ should not invalidate the coloring. For vertices in some set $T_j$, we have a similar issue. Our choice of starting problem ensures that if some combination~$S_{i^*}, T_{j^*}$ corresponding to input~$X_{i^*,j^*}$ has a $2$-$3$-coloring, then the remaining sets~$T_j$ can be safely colored~$3$, since vertices in $S_{i^*}$ will use only two of the available colors. The key insight to ensure that vertices in the remaining $S_i$ can also be colored, is to split them into multiple copies that each have at most one neighbor in any $T_j$. There will be at most one vertex in the neighborhood of a copy that is colored using color $1$ or $2$, thereby we can always color it using the other available color. Finally, additional gadgets will ensure that in some $S_i$ all these copies get equal colors, and in some $T_j$ the vertices that correspond to a triangle in the inputs are properly colored as such. With this intuition, we give the construction.

\begin{theorem} \label{thm:threecoloring:lowerbound}
\ThreeColoring parameterized by the number of vertices $n$ does not have a generalized kernel of size $\Oh(n^{2-\varepsilon})$ for any $\varepsilon > 0$, unless \containment.
\end{theorem}
\begin{proof}
To prove this statement, we give a degree-$2$ cross-composition from \pname{$2$-$3$-Coloring with triangle split decomposition} to \pname{$3$-List Coloring} and then show how to change this instance into a \pname{$3$-Coloring} instance. We start by defining a polynomial equivalence relation $\eqvr$ on instances of \pname{$2$-$3$-Coloring with triangle split decomposition}. Let two instances be equivalent under \eqvr, when the sets $U$ have the same size and sets $V$ consist of the same number of triangles. It is easy to verify that \eqvr is a polynomial equivalence relation.

By duplicating one of the inputs several times if needed, we ensure that the number of inputs to the cross-composition is a square. This increases the number of inputs by at most a factor four and does not change the value of the \textsc{or}. Therefore, assume we are given $t$ instances of \pname{$2$-$3$-Coloring with Triangle Split Decomposition} such that $\q := \sqrt{t}$ is integer. Enumerate these instances as $X_{i,j}$ for $i,j \in [\q]$ and let instance $X_{i,j}$ have graph $G_{i,j}$. For input instance $X_{i,j}$, let $U$ and $V$ be such that $U$ is an independent set with $|U| = m$ and $V$ consists of $n$ vertex-disjoint triangles. Enumerate the vertices in $U$ as $u_1,\ldots,u_m$ and in $V$ as $v_1,\ldots,v_{3n}$ such that $v_{3k-2},v_{3k-1},v_{3k}$ form a triangle for $k \in [n]$. We now create an instance of the \pname{$3$-List Coloring} problem, consisting of a  graph $G'$ together with a list function $L$ that assigns a subset of the color palette $\{1,2,3\}$ to each vertex.

Refer to Figure \ref{fig:4-coloring} for a sketch of $G'$.

\begin{figure}
\begin{minipage}[t]{.75\linewidth}
\centering\includegraphics{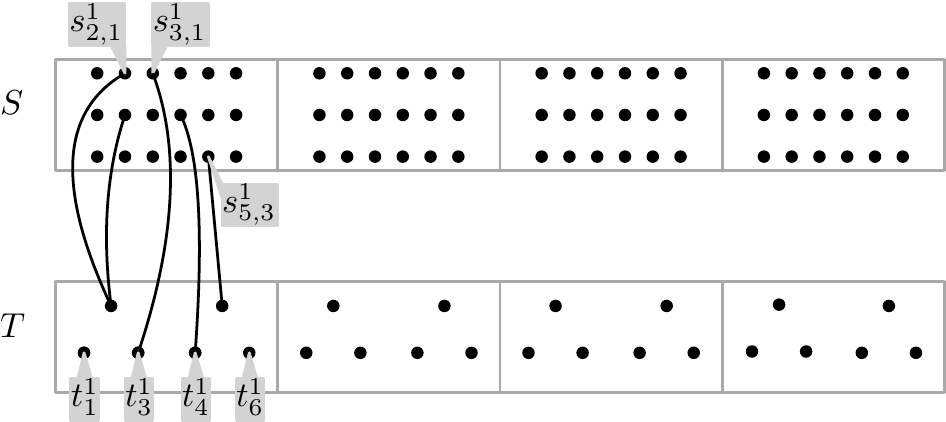}
\subcaption{Constructed graph $G'$}\label{fig:4-coloring_G}
\end{minipage}%
\quad
\begin{minipage}[t]{.2\linewidth}
\centering\includegraphics{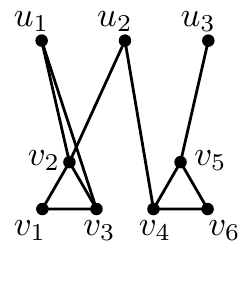}
\subcaption{Instance $X_{1,1}$}\label{fig:4-coloring_X}
\end{minipage}
\caption{Construction of graph $G'$ for $\q = 4$, $m = 3$, and $n=2$. Edges between vertices in $S$ and $T$ are shown for instance $X_{1,1}$. All \newgadgets and the vertex sets $A$ and $B$ are left out.}\label{fig:4-coloring}
\end{figure}

\begin{enumerate}
\item \label{step:S} Initialize~$G'$ as the graph containing $\q$ sets of $m\cdot 3n$ vertices each, called $S_i$ for $i~\in~[\q]$. Label the vertices in each of these sets as $s^i_{k,\ell}$ for $i \in [\q]$, $k \in [3n]$ and $\ell \in [m]$. Define $L(s^i_{k,\ell})~:=~\{1,2\}$.  The vertices $s^i_{1,\ell}, s^i_{2,\ell},\ldots ,s^i_{3n, \ell}$ together represent a single vertex of the independent set of an input instance, which is split into copies to ensure that every copy has at most one neighbor in each cell of $T$ (the bottom row in Figure \ref{fig:4-coloring_G}).
\item \label{step:T} Add $\q$ sets of $3n$ vertices each, labeled $T_j$ for $j \in [\q]$. Label the vertices in $T_j$ as $t_k^j$ for $k \in [3n]$ and let $L(t_k^j) := \{1,2,3\}$. Vertices $t_{3k-2}^j,t_{3k-1}^j,t_{3k}^j$ correspond to a triangle in an input graph. They are not connected, so that we can safely color all vertices that do not correspond to a $3$-colorable input with color $3$.
\item \label{step:connect} Connect vertex $s^i_{k,\ell}$ to vertex $t^j_k$ if in graph $G_{i,j}$ vertex $u_\ell$ is connected to $v_k$, for $k \in [3n]$ and $\ell \in [m]$. By this construction, the graph $G_{i,j}$ is isomorphic to the graph obtained from~$G'[S_i \cup T_j]$ by replacing each triple $t^j_{3k-2}, t^j_{3k-1},t^j_{3k}$ by a triangle for $k \in [n]$ and merging all $3n$ vertices $s^i_{k,\ell}$ in $S_i$ that have the same value for $\ell \in [m]$.
\item \label{step:buttons} Add vertex sets $A = \{a_1,\ldots,a_{\q}\}$ and $B:= \{b_1,\ldots,b_{\q}\}$. These are used to choose indices $i$ and $j$ such that $G_{i,j}$ is $3$-colorable. Let $L(a_i) := L(b_i) := \{1,2\}$ for all $i \in [\q]$.
\item \label{step:gadgets_buttons_A} Let $\vect{c}$ be defined by $c_i := 2$ for all $i \in [\q]$. Add a $\newgadget(\vect{c})$ to $G'$. Connect vertex $a_i$ to the distinguished vertex $\pi_i$ of this \newgadget for all $i \in [\q]$.
\item \label{step:gadgets_buttons_B} Let $\vect{c}$ again be defined by $c_i := 2$ for all $i \in [\q]$. Add a $\newgadget(\vect{c})$ to $G'$. Connect vertex $b_j$ to $\pi_j$ for all $j \in [\q]$. Together with the previous step, this ensures that in any proper list coloring at least one vertex in $A$ and at least one vertex in~$B$ has color $1$.
\item  \label{step:gadgets_S} For every $i \in [\q]$, $\ell \in [m]$, and $k \in [3n-1]$, for every  $c_1,c_2 \in [2]$ with $c_1 \neq c_2$, add a $\newgadget((c_1,c_2,1))$ to $G'$. Connect  $s_{k,\ell}^i$ to $\pi_1$, $s_{k+1,\ell}^i$ to $\pi_2$, and $a_i$ to $\pi_3$. This ensures that when $a_i$ has color $1$, vertices $s^i_{k,\ell}$ and $s^i_{k',\ell}$ have the same color for all $k,k' \in [3n]$.
\item \label{step:gadgets_T} For every $j \in [\q]$, $k \in [n]$, for every $c_1,c_2,c_3 \in [3]$ that are not all pairwise distinct, add a $\newgadget((c_1,c_2,c_3,1))$ to $G'$. Connect $t^j_{3k-2}$ to $\pi_1$, $t^j_{3k-1}$ to $\pi_2$, $t^j_{3k}$ to $\pi_3$, and $b_j$ to $\pi_4$. This construction ensures that if $b_j$ is colored $1$, all ``triangles'' in $T_j$ are properly colored. If $b_j$ is colored $2$ however, the gadgets add no additional restrictions to the coloring of vertices in $T_j$.
\end{enumerate}

\noindent This concludes the construction of~$G'$; we proceed with the analysis.

\begin{myclaim}\label{claim:existsS}
Let $c$ be a proper $3$-list coloring of $G'$. Then there exists $i \in [\q]$ such that for all $\ell \in [m]$ and for all $k,k' \in [3n]$ we have $c(s_{k,\ell}^i) = c(s_{k',\ell}^i)$.
\end{myclaim}
\begin{claimproof}
By the \newgadget added in Step \ref{step:gadgets_buttons_A}, there exists $i \in [\q]$ such that $c(a_i) \neq 2$. Since $L(a_i) = \{1,2\}$, this implies that $c(a_i) = 1$. We show that $i$ has the required property.

Suppose there exist $k,k' \in [3n]$ and $\ell \in [m]$ such that $c(s_{k,\ell}^i) \neq  c(s_{k',\ell}^i)$. Then there must also exist $k \in [3n-1]$ such that $c(s_{k,\ell}^i) \neq  c(s_{k+1,\ell}^i)$, or else they would all be equal.
Let $(c_1,c_2,c_3)$ correspond to the coloring of $s_{k,\ell}^i$,$s_{k+1,\ell}^i$, and $a_i$ as given by $c$. Then $\newgadget((c_1,c_2,c_3))$ was added in Step \ref{step:gadgets_S} and connected to these three vertices. But by Lemma \ref{lem:blocking_gadget}, it follows that any list-coloring of this \newgadget must assign color $c_i$ to some $\pi_i$ for $i\in[3]$. By the way they are connected to $s_{k,\ell}^i$,$s_{k+1,\ell}^i$ and $a_i$, one edge has two endpoints of equal color, which is a contradiction.
\end{claimproof}

We will say a triple of vertices $v_1,v_2,v_3$ is \emph{colorful} (under coloring $c$), if they receive distinct colors, meaning $c(v_1)\neq c(v_2)\neq c(v_3) \neq c(v_1)$.

\begin{myclaim}\label{claim:existsT}
Let $c$ be a proper $3$-list coloring of $G'$. Then there exists $j \in [\q]$ such that for all  $k \in [n]$  the triple $t_{3k}^j$, $t_{3k-1}^j$, $t_{3k-2}^j$ is colorful.
\end{myclaim}
\begin{claimproof}
By the \newgadget added in Step \ref{step:gadgets_buttons_B}, there exists $j\in[\q]$ such that $c(b_j) \neq 2$. Since $L(b_j) = \{1,2\}$, this implies that $c(b_j) = 1$. We show that $j$ has the desired property.

Suppose there exists $k \in [n]$, such that $t_{3k}^j$, $t_{3k-1}^j$, and $t_{3k-2}^j$ are not a colorful triple. Let $(c_1,c_2,c_3,c_4) \in [3]^4$ correspond to the coloring given to $t_{3k}^j$, $t_{3k-1}^j$, $t_{3k-1}^j$, and  $b_j$. In Step \ref{step:gadgets_T}, $\newgadget((c_1,c_2,c_3,c_4))$ was added, together with connections to these four vertices.
But by Lemma \ref{lem:blocking_gadget}, any list-coloring of this \newgadget must assign color $c_i$ to some $\pi_i$ for $i\in [4]$. By the way they are connected to~$t_{3k}^j$, $t_{3k-1}^j$, $t_{3k-2}^j$, and $b_j$, one edge has two endpoints of equal color, which is a contradiction.
\end{claimproof}

\begin{myclaim}\label{claim:cross}
The graph $G'$ is $3$-list colorable $\Leftrightarrow$ some input instance $X_{i^*j^*}$ is $2$-$3$-colorable.
\end{myclaim}
\begin{claimproof}
$(\Rightarrow)$
Suppose we are given a $3$-list coloring $c$ of $G'$. By Lemmas \ref{claim:existsS} and \ref{claim:existsT} there exist integers ${i^*}$ and ${j^*} \in [\q]$ such that for all $\ell \in [m]$ and for all $k,k'\in[3n]$ we have $c(s^{i^*}_{k,\ell}) = c(s^{i^*}_{k',\ell})$ and furthermore for all $k \in [n]$ the triple $t_{3k}^{j^*}$, $t_{3k-1}^{j^*}$, $t_{3k-2}^{j^*}$ is colorful. We show that this implies that $G_{i^*,j^*}$ has a valid $2$-$3$-coloring $c'$, which we define as follows. Let $c'(u_\ell):= c(s^{i^*}_{1,\ell})$ for $\ell \in [m]$ and let $c'(v_k):= c(t^{j^*}_k)$ for $k \in [3n]$. It remains to verify that $c'$ is a valid coloring of $G_{i^*,j^*}$. For any edge $\{u_\ell,v_k\}\in E(G_{i^*,j^*})$ with $\ell\in[m],k\in[3n]$, the endpoints receive different colors since
\[ c'(v_k) = c(t^{j^*}_{k})\neq c(s^{i^*}_{k,\ell}) = c(s^{i^*}_{1,\ell}) = c'(u_\ell).\]
For an edge $\{v_k,v_k'\} \in G_{i^*,j^*}$, its coloring corresponds to the coloring of $t_k^{j^*}$ and $t_{k'}^{j^*}$, which are colored differently by choice of ${j^*}$ in Lemma \ref{claim:existsT}. Furthermore, $u_\ell$ is always colored with color $1$ or $2$ as $L(s^{i^*}_{1,\ell}) = \{1,2\}$. Thereby, $c'$ is a proper $2$-$3$-coloring of $G_{i^*,j^*}$.

$(\Leftarrow)$ Suppose $c$ is a $2$-$3$-coloring of $G_{i^*,j^*}$, such that the $U$-partite set of $G_{i^*,j^*}$ is colored using only the colors $1$ and $2$. We will construct a $3$-list coloring $c'$ for graph $G'$. For $\ell \in [m]$ let $c'(s_{k,\ell}^{i^*}) := c(u_\ell)$ for all $k \in [3n]$. For $k \in [3n]$  let $c'(t^{j^*}_k) := c(v_k)$. For $j \neq j^*$ and $k \in [3n]$ let $c'(t^j_k) :=  3$. For $i \neq i^* \in [t']$, $k \in [3n]$ and $\ell \in [m]$, pick $c'(s^i_{k,\ell}) \in \{1,2\} \setminus \{c'(t^{j^*}_k)\}$. Let $c'(a_{i^*}):= 1$ and let $c'(b_{j^*}) := 1$. For $i \neq i^*$ let $c'(a_i) := 2$, similarly for $j \neq j^*$ let $c'(b_j):= 2$.  Before coloring the vertices in \newgadgets, we will show that $c'$ is proper on $G'[S \cup T]$. This will imply that the coloring defined so far is proper, as vertices in $A$ and $B$ only connect to \newgadgets.

Note that all edges in $G'[S \cup T]$ go from $S$ to $T$. Consider an edge $\{s,t\}$ for $s \in S, t \in T$.
Since $c'(s) \neq 3$, if $t \in T_j$ for $j \neq j^* \in [t']$, it follows immediately that $c'(s)\neq c'(t)$. Furthermore, if $s \in S_i$ for $i\neq i^* \in [t']$, $c'(s) \neq c'(t)$ by the definition of $c'(s)$. Otherwise, $s \in S_{i^*}$ and $t \in T_{j^*}$ and there exist $\{u,v\} \in E(G_{i^*,j^*})$ such that $c'(s) = c(u)$ and $c'(t) = c(v)$. Since $c$ is a proper coloring, it follows that $c'(s)\neq c'(t)$.

To complete the proof, extend $c'$ to also properly color all \newgadgets. This is possible for the \newgadgets added in Steps \ref{step:gadgets_buttons_A} and \ref{step:gadgets_buttons_B}, since $c'(a_{i^*}) = 1$ and $c'(b_{j^*}) = 1$. Furthermore, we show  this is possible for all \newgadgets introduced in Step \ref{step:gadgets_S}. A $\newgadget((c_1,c_2,c_3))$ introduced in Step \ref{step:gadgets_S} either has $\pi_3$ connected to $a_i$ for  $i \neq i^*$ with $c'(a_i) = 2 \neq c_3$, or it is connected to $a_{i^*}$ and in this case the vertices $s^{i^*}_{k,\ell}$ and $s^{i^*}_{k+1,\ell}$ are assigned equal colors and thus at least one of them has a coloring different from the coloring given by $c_1$ and $c_2$ as these colors are distinct. Thus, the colors that are forbidden on vertices $\pi_i$ by the connections to the rest of the graph, do not correspond to $(c_1,c_2,c_3)$ and $c'$ can be extended to color the entire \newgadget by Lemma \ref{lem:blocking_gadget}.

Similarly, coloring $c'$ can be extended to $\newgadgets(\vect{c})$ added in Step \ref{step:gadgets_T}, as either $\pi_4$ in the gadget is connected to $b_j$ for $j\neq j^*$ and  $c(b_j) = 2 \neq c_4$, or the three  vertices from $T$ connected to this gadget are colored with three different colors.
\end{claimproof}

The claim above shows that we have given a cross-composition into \textsc{$3$-List Coloring}. To obtain an instance of \textsc{$3$-Coloring}, we add a triangle consisting of vertices $\{C_1,C_2,C_3\}$ to the graph. We connect a vertex $v$ in $G'$ to $C_i$ if $i \notin L(v)$ for $i\in [3]$. This graph now has a proper $3$-coloring if and only if the original graph had a proper $3$-list coloring. Thus, by Claim \ref{claim:cross}, the resulting \pname{$3$-Coloring} instance acts as the logical \textsc{or} of the inputs.

It remains to bound the  number of vertices of $G'$.
In Step \ref{step:S} we add $|S| = m\cdot 3n\cdot \q$ vertices and in Step \ref{step:T} we add another $|T| = 3n \cdot \q$ vertices.
Then in Step \ref{step:buttons} we add $|A|+ |B| = 2\q$ additional vertices.
The two \newgadgets added in Steps \ref{step:gadgets_buttons_A} and \ref{step:gadgets_buttons_B} each have size $\Oh(\q)$.
The \newgadgets added in Step \ref{step:gadgets_S} have constant size, and we add six of them for each $i \in [\q], \ell \in [m], k \in [3n-1]$, thus adding $\Oh(\q\cdot m \cdot n)$ vertices.
Similarly, the \newgadgets added in Step \ref{step:gadgets_T} have constant size, and we add a constant number of them for each $j \in [\q], \ell \in [n]$, thus adding $\Oh(\q\cdot n)$ vertices. This gives a total of $\Oh(\q \cdot n \cdot m) = \Oh(\sqrt{t}\cdot (\max_{i,j} |X_{i,j}|)^{\Oh(1)})$ vertices. Theorem~\ref{thm:threecoloring:lowerbound} now follows from Theorem \ref{thm:cross_composition_LB} and Lemma~\ref{lem:NP_complete_coloring}.
\end{proof}

The set of all vertices of a graph is always a valid vertex cover for that graph. Thereby, it follows from Theorem \ref{thm:threecoloring:lowerbound} that the  lower bound also holds when parameterized by vertex cover. In \cite[Theorem 3]{JansenK2013Data}, it was shown that for any $q \geq 4$, \pname{$q$-Coloring} parameterized by vertex cover does not have a generalized kernel of size $\Oh(k^{q - 1 - \varepsilon})$, unless \containment. Combining these results gives a lower bound for \pname{$q$-Coloring} parameterized by vertex cover size.

\begin{corollary}
For any $q \geq 3$,  \pname{$q$-Coloring}  parameterized by vertex cover does not have a generalized kernel of bitsize $\Oh(k^{q - 1 -\varepsilon})$ for any $\varepsilon > 0$, unless \containment.
\end{corollary}

The above lower bound carries over to \pname{$q$-Coloring} parameterized by the size of a twin-cover, since any vertex cover of a graph is also a valid twin-cover. Recall that \pname{$q$-Coloring}  is equivalent to \pname{$K_q$-Coloring} and $\Delta(K_q) = q-1$. Thereby, the above lower bound matches the kernel size given in Theorem \ref{thm:kernel} up to $k^{o(1)}$ factors.
\section{Conclusion}
We have given a kernel for \pname{$H$-Coloring parameterized by Twin-Cover} with $\Oh(k^{\Delta(H)})$ vertices and bitsize $\Oh(k^{\Delta(H)}\log{k})$. This kernel can be obtained without using information about (an approximation of) the minimum twin-cover of the input graph. It follows from this result that \pname{$q$-Coloring parameterized by Vertex Cover} has a kernel of bitsize $\Oh(k^{q-1}\log k)$, improving on the previously known kernel by almost a factor $k$. Furthermore, \pname{$3$-Coloring} when parameterized by the number of vertices has no kernel of size $\Oh(n^{2-\varepsilon})$, unless \containment. It was already known that for $q \geq 4$, \pname{$q$-Coloring parameterized by Vertex Cover} was unlikely to have a kernel of size $\Oh(k^{q-1-\varepsilon})$. Combining these results allows us to give the same lower bound for $q=3$, under the assumption that \notcontainment. Thereby we have provided an upper and lower bound on the kernel size of \pname{$q$-Coloring parameterized by Vertex Cover} for any $q \geq 3$, that match up to $k^{o(1)}$ factors.

It is easy to see that the kernel lower bounds also hold for \pname{$q$-List Coloring}, where every vertex $v$ in the graph has a list $L(v) \subseteq [q]$ of allowed colors. Furthermore, we can also apply our kernel, by first reducing an instance of \pname{$q$-List Coloring} to an instance of \pname{$q$-Coloring} using $q$ additional vertices, and adding these $q$ vertices to the twin-cover of the graph. This only changes the size of the obtained kernel by a constant factor. The kernel does not extend to the general \pname{List $H$-Coloring} problem, since the gadget to simulate the list constraints only works correctly when $H$ is a clique.

In this paper we gave a first example where finding redundant vertices and edges is done using appropriate polynomial equalities. It would be interesting to see if this technique  can be applied to obtain smaller kernels for other graph problems as well. To apply this idea, one needs to first identify which constraints should be modeled. When the constraints are found, they need to be written as equalities of low-degree polynomials over a suitably chosen field. This requires the clever construction of polynomials that have a sufficiently low degree, in order to obtain a good bound on the kernel size.

\subparagraph*{Acknowledgements}
We thank Tim Hartmann for suggesting the use of twin-cover as a parameter.

\bibliographystyle{plainurl}    
\bibliography{references}   

%
%

\end{document}